\documentclass[conference]{IEEEtran}
\IEEEoverridecommandlockouts
\usepackage{amsmath,amsfonts}
\usepackage{mathtools}
\usepackage{amsthm}
\usepackage{array}
\usepackage{textcomp}
\usepackage{stfloats}
\usepackage{url}
\usepackage{microtype}
\usepackage{verbatim}
\usepackage{graphicx}
\usepackage{cite}
\usepackage[T1]{fontenc}
\usepackage[colorlinks=true, linkcolor=blue, citecolor=blue, urlcolor=blue]{hyperref}
\usepackage{enumitem}

\usepackage[linesnumbered,lined,algoruled,commentsnumbered]{algorithm2e}

\usepackage{caption}
\usepackage{subcaption}
\usepackage{stfloats}
\usepackage{lipsum} 
\usepackage[font=small,skip=0pt]{caption}
\usepackage[export]{adjustbox}
\usepackage[inkscapelatex=false]{svg}

\SetCommentSty{mycommfont}
\theoremstyle{definition}
\newtheorem{theorem}{Theorem}

\usepackage{units}
\newcommand{\nth}[1]{{#1}^{\text{th}}}

\newcommand{\NBS}[0]{N_{\mathrm{\scalebox{0.6} {BS} }}}

\newcommand{\RD}[0]{r_{\mathrm{\scalebox{0.6} {RD} }}}
\newcommand{\NRF}[0]{N_{\mathrm{\scalebox{0.6} {RF}}}}
\newcommand{\rf}[0]{r_{\mathrm{\scalebox{0.6} {F}}}}
\newcommand{\BD}[0]{r_{\mathrm{\scalebox{0.6} {BD} }}}
\newcommand{\EBRD}[0]{r_{\mathrm{\scalebox{0.6} {EBRD} }}}

\newcommand{\bl}[0]{\mathbf{b}_\mathrm{\scalebox{0.6} {ULA}}}
\newcommand{\br}[0]{\mathbf{b}_\mathrm{\scalebox{0.6} {URA}}}
\newcommand{\GR}[0]{\mathcal{G}_\mathrm{\scalebox{0.5} {URA}}}
\newcommand{\MUE}[0]{M_{\mathrm{UE}} }

\usepackage{acronym}
\input{acronym}

\begin{document}
\title{Analyzing URA Geometry for Enhanced Spatial Multiplexing and Extended Near-Field Coverage 
\author{Ahmed Hussain, Asmaa Abdallah, Abdulkadir Celik, and Ahmed M. Eltawil,
\\ Computer, Electrical, and Mathematical Sciences and Engineering (CEMSE) Division,
\\King Abdullah University of Science and Technology (KAUST), Thuwal, 23955-6900, KSA }
}

\maketitle

\begin{abstract}

With the deployment of large antenna arrays at high-frequency bands, future wireless communication systems are likely to operate in the radiative near-field. Unlike far-field beam steering, near-field beams can be focused within a spatial region of finite depth, enabling spatial multiplexing in both the angular and range dimensions. This paper derives the beamdepth for a generalized uniform rectangular array (URA) and investigates how array geometry influences the near-field beamdepth and the limits where near-field beamfocusing is achievable. To characterize the near-field boundary in terms of beamfocusing and spatial multiplexing gains, we define the effective beamfocusing Rayleigh distance (EBRD) for a generalized URA. Our analysis reveals that while a square URA achieves the narrowest beamdepth, the EBRD is maximized for a wide or tall URA. However, despite its narrow beamdepth, a square URA may experience a reduction in multiuser sum rate due to its severely constrained EBRD. Simulation results confirm that a wide or tall URA achieves a sum rate of
$3.5 \times$ more than that of 
a square URA, benefiting from the extended EBRD and improved spatial multiplexing capabilities.

\end{abstract}



\section{Introduction}
Massive \ac{MIMO} has been a cornerstone of \ac{5G} advancements, revolutionizing wireless communication through unprecedented spatial multiplexing gains \cite{saad2019vision}. As wireless networks continue to evolve, next-generation systems are poised to embrace \ac{UM}-\ac{MIMO}, leveraging even larger antenna arrays and expanding into higher frequency spectra to meet the growing demands for capacity and efficiency \cite{Jiang2021Survey}. 
At higher carrier frequencies, the reduced wavelength enables the deployment of massive antenna arrays within a confined space. The electromagnetic field surrounding an antenna array is generally divided into three zones: the reactive near-field, the radiative near-field, and the far-field. Extending from the reactive near-field up to the Rayleigh distance, the radiative near-field (also called the Fresnel region) is dominated by spherical wave propagation. Notably, the Rayleigh distance, which marks the boundary between the radiative near-field and the far-field, increases significantly with a larger array aperture and a shorter wavelength \cite{cui2022near}. As a result, upcoming wireless technologies are expected to operate predominantly in the radiative near-field. 

Recent studies highlight that considering near-field propagation effects can markedly boost communication capacity \cite{10273772}. As wireless systems transition from the far-field to the near-field regime, MIMO technology experiences a significant enhancement in spatial multiplexing gains \cite{miller2019waves}. Furthermore, the near-field \ac{LoS} channel, influenced by spherical wave propagation, can foster a rich scattering environment, thereby increasing single-user capacity. Although spherical wavefront propagation contributes to capacity improvements, an in-depth analysis is crucial to better understand the distinct electromagnetic properties of the near-field and further optimize communication performance. 

\begin{figure}[t]
\centering
\includegraphics[width = 0.95\columnwidth]{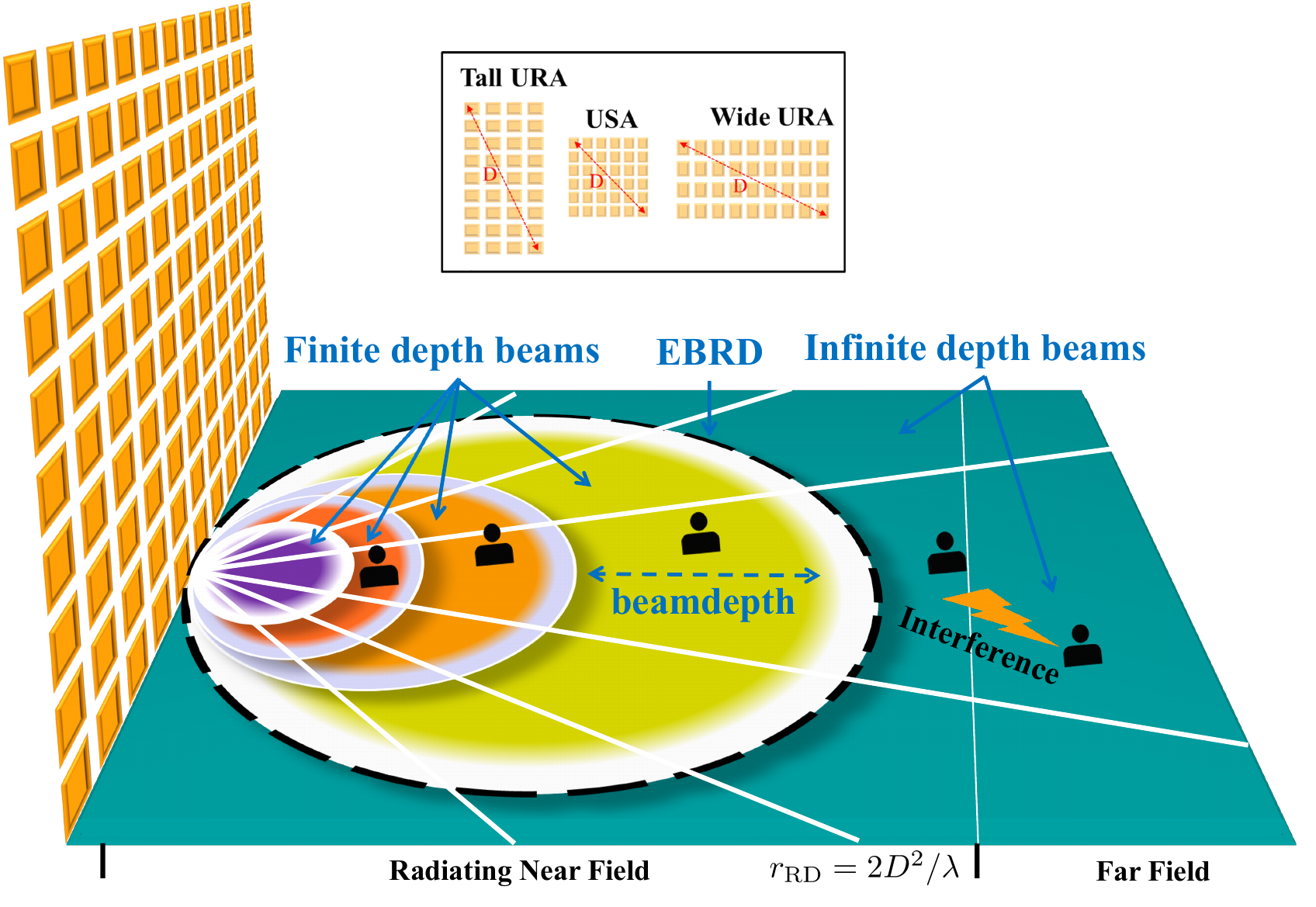}
\caption{Near-field beams with finite depth are achieved within EBRD of near-field region.}
\label{fig1_motivation}
\end{figure}
 
In the far-field, RF beams are characterized by a limited angular beamwidth while possessing an infinite beamdepth. Conversely, as depicted in Fig. \ref{fig1_motivation}, near-field beams can achieve a finite beamdepth, allowing beamfocusing at designated regions. This focusing ability requires compensating for phase shifts arising from the varying distances between antenna elements and the targeted focal point, leading to spherical equiphase surfaces converging at the intended location. The presence of a finite beamdepth also enables the formation of multiple beams at different ranges within the same angular direction, improving spatial multiplexing. Notably, the \ac{LDMA} approach proposed in \cite{wu2023multiple} has demonstrated superior spectral efficiency compared to \ac{SDMA}.

Primal work on near-field beamforming \cite{bjornson2021primer}, derived beamdepth expressions for broadside square aperture antennas, which were later extended to rectangular array configurations in \cite{10443535}. To better characterize wideband near-field channels, the work in \cite{10541333} introduced the \ac{ERD} metric, which quantifies far-field beamforming loss within the near-field domain. This concept was further refined for wideband phased arrays using the \ac{BAND} framework \cite{deshpande2022wideband}. Additionally, for a given aperture size, \acp{UCA} encompass a larger near-field region than \acp{ULA} due to their rotational symmetry, though they exhibit limitations in broadside focusing.

The conventional Rayleigh distance tends to overestimate the actual boundaries of the near-field region \cite{hussain2024near,hussain2025nearBT}. For example, alternative metrics like \ac{ERD} and \ac{BAND} cover only a small region of the near-field area defined by the Rayleigh distance. However, since these measures are based on far-field beamforming loss, they do not fully capture essential near-field characteristics such as beamfocusing, polar-domain sparsity, and spatial \ac{DoF}. To address these shortcomings, the \ac{EBRD} was introduced, as depicted in Fig. \ref{fig1_motivation}, to better define beamfocusing limits and polar-domain sparsity for a \ac{ULA} \cite{10988573}.

Recent studies distinguish between \ac{CAP} antennas, which feature continuous radiating surfaces, and \ac{SPD} antennas, composed of discrete radiating elements \cite{liu2023near}. While \ac{CAP} antennas generally produce well-formed radiation patterns, \ac{SPD} antennas—especially in large arrays—tend to generate stronger sidelobes and increased interference. Notably, the beamdepth analyses in \cite{bjornson2021primer, 10443535} primarily focus on \ac{CAP} antennas, assuming a fixed Rayleigh distance irrespective of the array aperture length. However, for \ac{SPD} antennas, the Rayleigh distance varies in proportion to the aperture length. Additionally, most prior research emphasizes broadside transmission, underscoring the need for angle-dependent studies on beamdepth and near-field beamfocusing for \ac{SPD} antennas.

Given these considerations, further research is required to identify array geometries that achieve narrow beamdepth while extending the \ac{EBRD}, enabling more users to operate within the effective near-field region. This work addresses these gaps by analyzing near-field beamfocusing for an \ac{SPD} \ac{URA} with a variable width-to-height ratio. Specifically, we derive the angle-dependent beamdepth, establish beamdepth limits referred to as the \ac{EBRD} for a generalized \ac{URA}, and analyze the impact of array geometry on these characteristics. In contrast to \cite{10443535}, which assumes a fixed Rayleigh distance, we model it as a function of the array aperture and derive angle-dependent formulations for near-field beamfocusing. Our analysis indicates that while a square URA yields the narrowest beamdepth, the EBRD is maximized with a wide or tall URA. However, despite its narrower beamdepth, a square URA may suffer from a reduced multiuser sum rate due to its significantly constrained EBRD. Simulation results demonstrate that a wide or tall URA achieves a sum rate of $\unit[8.2]{bps/Hz}$, compared to only $\unit[2.3]{bps/Hz}$ for a square URA, owing to its extended EBRD and enhanced spatial multiplexing capabilities.

The rest of the paper is organized as follows: Section \ref{Sec_System_Model} introduces the near-field channel model for a generalized \ac{URA}. Section \ref{sec-3} focuses on deriving the beamdepth and \ac{EBRD}, along with an analysis of how array geometry influences these metrics. Section \ref{sec-4} presents simulation results, while Section \ref{sec-5} concludes the paper.


\section{System Model} \label{Sec_System_Model}
We consider a narrowband communication system featuring a \ac{UM}-\ac{MIMO} antenna array at the \ac{BS} and a single-antenna \ac{UE}, as illustrated in Fig. \ref{fig1_system_model}. The \ac{BS} employs a \ac{URA} comprising $N_\text{BS}$ antenna elements, arranged in a grid with $N_1$ elements along the y-axis and $N_2$ along the z-axis. The antenna elements are uniformly spaced in both directions, with a separation of $d_y = d_z = d = \frac{\lambda}{2}$. The width-to-height ratio of the array is given by $\eta = \frac{N_1}{N_2}$. 
The indexing for antenna elements along the y-axis is defined as $n_1 \in [-\overset{\sim}{N}_{1}, \dots, 0, \dots, \overset{\sim}{N}_{1}]$, where $\overset{\sim}{N}_{1} = \left \lceil \frac{N_1-1}{2} \right\rceil$. Similarly, the elements along the z-axis follow the indexing $n_2 \in [-\overset{\sim}{N}_{2}, \dots, 0, \dots, \overset{\sim}{N}_{2}]$, with $\overset{\sim}{N}_{2} = \left \lceil \frac{N_2-1}{2} \right\rceil$. Using the Pythagorean theorem, the total aperture length of the \ac{URA} is determined as $D = d \sqrt{N_1^2 + N_2^2}$. The Rayleigh distance is defined as $\RD = \frac{2D^2}{\lambda}$ \cite{hussain2025near}. The \ac{UE} is positioned at a distance $r$ from the \ac{BS}, forming an azimuth angle $\varphi$ and an elevation angle $\theta$. To ensure minimal amplitude variation across the array, the minimum distance within the radiative near-field is set to $1.2D$. This allows the use of a \ac{USW} model for subsequent analysis \cite{liu2023near}.

\begin{figure}[t]
\centering
\includegraphics[width = 0.9\linewidth]{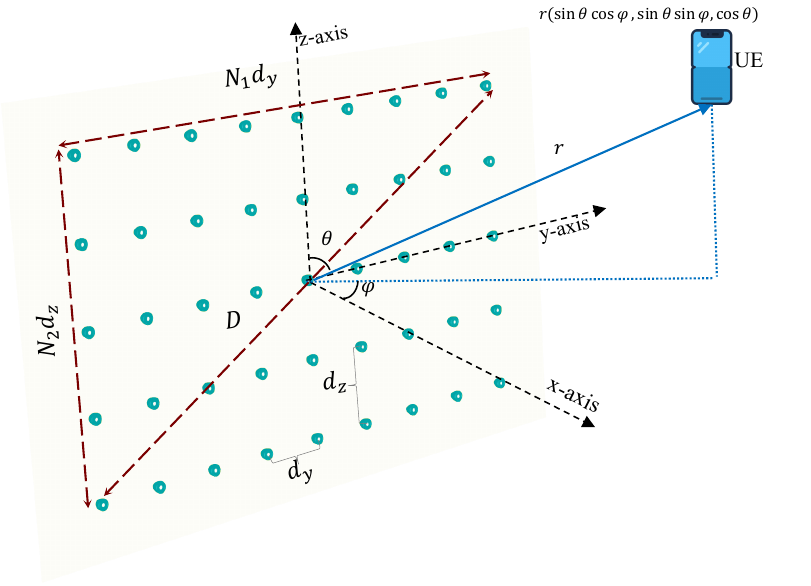}
\caption{Near-field channel model featuring UM-MIMO \ac{URA} array at the BS and single antenna UE.}
\label{fig1_system_model}
\end{figure}
The near-field channel $\mathbf{h} \in \mathcal{C}^{\NBS\times 1}$ based on the \ac{USW} model is mathematically formulated as \cite{sherman1962properties}
\begin{equation} 
\mathbf{h} = \sqrt{\frac{\NBS}{L}}\sum_{l=0}^{L}g_le^{-j{\nu}r_l}\mathbf{b}(\varphi_l,\theta_l,r_l),
\label{eqn_B1}
\end{equation}
here, $\nu = \frac{2\pi{f}}{c}$ represents the wavenumber, where $f$ denotes the carrier frequency. The term $g_l$ corresponds to the complex gain associated with the $\nth{l}$ propagation path among a total of $L$ paths, while $r_l$ represents the distance between the $\nth{l}$ scatterer and the \ac{URA}. Furthermore, $\mathbf{b}(\varphi_l, \theta_l, r_l)$ denotes the near-field steering vector, which directs energy within the spatial region specified by the angles $(\varphi_l, \theta_l)$ and the distance $r_l$. The near-field array response vector for a \ac{ULA} consisting of $N_1$ antenna elements is expressed as
\small
\begin{equation} 
\bl(\varphi_l,r_l) = \tfrac{1}{\sqrt{N_1}}\Big[e^{-j\nu(r_l^{(-\overset{\sim}{N}_{1})} -r_l)},\dots, e^{-j\nu(r_l^{(\overset{\sim}{N}_{1})} -r_l)}\Big],
\label{eqn_B2}
\end{equation}
\normalsize
where $r_{l}^{(n)}$ represents the distance between the $\nth{l}$ scatterer and the $\nth{n}$ antenna element. For a \ac{URA} with $\NBS = N_1 N_2$ elements, the near-field beamfocusing vector can be derived using the spherical-wave propagation model as
\small
\begin{equation}
\begin{aligned}
\br(\varphi_l,\theta_l,r_l)
=\frac{1}{\sqrt{\NBS}}\left[e^{-j {
\nu}(r_l^{(-\tilde{\boldsymbol{\zeta}})}-r_l)}, \cdots, e^{-j {
\nu}(r_l^{(\tilde{\boldsymbol{\zeta}})}-r_l)}\right],
\label{eqn_B3}
\end{aligned}
\end{equation}
\normalsize
where $\widetilde{\boldsymbol{\zeta}} = (\overset{\sim}{N_1}, \overset{\sim}{N_2})$. The distance between the $\nth{l}$ scatterer and the $\nth{(n_1, n_2)}$ element of the \ac{URA} is 
\small
\begin{equation}
\begin{aligned}
r_{l}^{\left(n_{1}, n_{2}\right)} & =\sqrt{\left(r_l u_x-0 \right)^{2}+\left(r_l u_y - n_1 d\right)^{2}+\left(r_l u_z - n_2 d\right)^{2}} \\
& \stackrel{(a)}{\approx} r_l - n_1 d u_y - n_2 d u_z + \frac{n_1^{2} d^{2}}{2r_l} \beta_1 \cdots \\ & + \frac{n_2^{2} d^{2}}{2r_l} \beta_2 - \frac{n_1 n_2 d^{2} u_y u_z}{r_l}, 
\label{eqn_B4}
\end{aligned}
\end{equation}\normalsize
where the directional cosines are given by: $u_x = \sin \theta \cos \varphi$, $u_y = \sin \theta \sin \varphi$, and $u_z = \cos \theta$. The approximation (a), also known as the \textit{near-field expansion}, is derived from the second-order Taylor series expansion of $\sqrt{1+x} \approx 1 + \frac{x}{2} - \frac{1}{8}x^2$ for small $x$. Moreover, we have $\beta_1 = 1 - \sin^2 \theta \sin^2 \varphi$ and $\beta_2 = 1 - \cos^2 \theta$. We neglect the last term of \eqref{eqn_B4} in the following analysis, as it becomes negligible for large $\NBS$ values \cite{10403506}, contributing only around $5\%$ to the array gain \cite{wu2023multiple}.

Finite-depth beamforming in the radiative near-field can be realized through the conjugate phase approach. This method involves adjusting the phase of each radiating element to counteract the phase delays between the focal point and individual antenna elements, thereby achieving constructive interference at the desired location. The generated beam is characterized by its azimuth beamwidth $\varphi_{\mathrm{\scalebox{0.5}{BW}}}$, elevation beamwidth $\theta_{\mathrm{\scalebox{0.5}{BW}}}$, and beamdepth $\BD$.

We consider a \ac{UM}-\ac{MIMO} \ac{BS} equipped with a generalized \ac{URA} that serves a near-field \ac{LoS} user located at $(\varphi, \theta, \rf)$. Let $\mathbf{w}(\varphi, \theta, \rf)$ represent the beamforming vector based on the conjugate phase method. The array gain $\GR$, which characterizes the normalized received power, is defined as
\begin{equation}
\begin{aligned}
&\GR = \left| \mathbf{w}\left(\varphi,\theta,z \right)^\mathsf{H} \br\left(\varphi,\theta,\rf\right) \right|^2, \\
&\forall \ \varphi \in [-\pi/2 ,\ \pi/2],\ \theta \in [-\pi/2 ,\ \pi/2], \ z \in [1.2D ,\ \infty]. 
\label{eqn_III}
\end{aligned}
\end{equation}

In the following analysis, we apply \eqref{eqn_III} to derive the array gain for the \ac{URA} and to characterize the beamdepth and \ac{EBRD} for various array geometries.


\section{Near-field Beamforming for URA} \label{sec-3}
In this section, we derive the beamdepth and the maximum range limits, known as the \ac{EBRD}, within which finite-depth beamforming is feasible for a \ac{URA}. We also explore how the spatial focus region and array geometry influence the variation in beamdepth and the extent of the \ac{EBRD}.

\subsection{Finite Beamdepth and EBRD} \label{sec-3a}
We define the beamdepth $\BD$, as the distance interval $z \in [\rf^\mathrm{min}, \rf^\mathrm{max}]$ within which the normalized array gain is at most $\unit[3]{dB}$ below its peak value. The beamforming vector $\mathbf{w}(\varphi, \theta, z)$ in (\ref{eqn_III}), created using the conjugate phase method, reaches its maximum array gain at the focal point $z=\rf$ and diminishes as the distance moves away from $\rf$. The normalized array gain achieved by $\mathbf{w}(\varphi, \theta, z)$ is given by
\begin{equation}
\GR \approx \left|\frac{1}{N_1 N_2} \sum_{n_{1}} \sum_{n_{2}} e^{j 
\frac{\pi}{\lambda}\left(n_{1}^{2} d^{2} \beta_1 + n_{2}^{2} d^{2} \beta_2 \right)
 z_\mathrm{eff} }\right|^2, 
\label{eqn_IIIB_1}
\end{equation}
where $ z_\mathrm{eff} = \left| \frac{z-\rf}{z \rf } \right| $. Given that $\NBS$ is large in \ac{UM}-\ac{MIMO}, the summation in \eqref{eqn_IIIB_1} can be approximated by a Riemann integral, yielding
\small
\begin{equation}
\begin{aligned}
\GR \approx \left|\frac{1}{N_1 N_2} \int_{-\tfrac{N_1}{2}}^{\tfrac{N_1}{2}} \int_{-\tfrac{N_2}{2}}^{\tfrac{N_2}{2}} e^{j 
\frac{\pi}{\lambda}\left(n_{1}^{2} d^{2} \beta_1 + n_{2}^{2} d^{2} \beta_2 \right)
 z_\mathrm{eff} } \mathrm{~d} n_{1} \mathrm{~d} n_{2} \right|^2.
 \label{eqn_IIIB_2}
\end{aligned}
\end{equation}
\normalsize
By changing variables, let $x_1 = \sqrt{\frac{2n_{1}^{2} d^{2} \beta_1 z_\mathrm{eff}}{\lambda}}$ and $x_2 = \sqrt{\frac{2n_{2}^{2} d^{2} \beta_2 z_\mathrm{eff}}{\lambda}}$, we can rewrite the array gain in \eqref{eqn_IIIB_2} in terms of Fresnel functions as:
\begin{equation}
\GR = \frac{\left[C^2(\gamma_1) + S^2(\gamma_1) \right]\left[C^2(\gamma_2) + S^2(\gamma_2)\right]}{(\gamma_1 \gamma_2)^2},
\label{eqn_IIIB_3}
\end{equation}
where $C(\gamma_i) = \int_{0}^{\gamma_i} \cos{\left(\frac{\pi}{2}x_i^2\right)}dx_i$ and $S(\gamma_i) = \int_{0}^{\gamma_i} \sin{\left(\frac{\pi}{2}x_i^2\right)}dx_i$ are the Fresnel functions, with $i \in \{1, 2\}$. Moreover, $\gamma_{1}=\sqrt{\frac{N_{1}^{2} d^{2} \beta_1 }{2\lambda}z_\mathrm{eff}}$ and $\gamma_{2}=\sqrt{\frac{N_{2}^{2} d^{2}\beta_2}{2\lambda}z_\mathrm{eff}}$.

\begin{theorem} 
For a generalized \ac{URA}, the $\unit[3]{dB}$ beamdepth $\BD$, achieved by focusing a beam at a distance $\rf$ from the \ac{BS}, is given by:
\begin{equation} 
\BD = \frac{8\rf^2 {\RD} \alpha_\mathrm{\scalebox{0.6}{3dB}} \ \eta (\eta^2+1)\sqrt{\beta_1 \beta_2}} {\left[\eta \RD \sqrt{\beta_1 \beta_2}\right]^2 -\left[4\alpha_\mathrm{\scalebox{0.6}{3dB}} \rf (\eta^2+1)\right]^2}.
\label{eqn_IIIB_4}
\end{equation}
\label{theorem1}
\end{theorem}
\begin{proof}
We define $\alpha_{\mathrm{\scalebox{0.6}{3dB}}} \stackrel{\Delta}{=}\left\{ \left(\gamma_1\gamma_2\right) \mid \GR\left(\gamma_1, \gamma_2 \right) = 0.5 \right\}$ Thus, $\alpha_\mathrm{\scalebox{0.6}{3dB}} = \frac{N_1 N_2 d^2}{2\lambda} \sqrt{ \beta_1 \beta_2} z_\mathrm{eff}$. Substituting $d^2 = D^2/(N_1^2 + N_2^2)$, $\alpha_\mathrm{\scalebox{0.6}{3dB}}$ can be obtained as
\begin{equation}
\begin {aligned}
\alpha_\mathrm{\scalebox{0.6}{3dB}}& = \frac{\RD N_1 N_2}{4(N_1^2 + N_2^2)}\sqrt{\beta_1 \beta_2} z_\mathrm{eff}\\
 &= \frac{\RD }{4} \left(\frac{\eta}{\eta^2+1}\right)\sqrt{ \beta_1 \beta_2} z_\mathrm{eff},
\end {aligned}
\label{eqn_IIIB_5}
\end{equation}
 where $\eta = N_1/N_2$. We can then solve for $z^{\star}$ in (\ref{eqn_IIIB_5}) to get $ z^{\star} = \frac{\rf \RD \sqrt{ \beta_1 \beta_2} \eta }{\RD \eta \sqrt{ \beta_1 \beta_2} \pm 4 \rf \alpha_\mathrm{\scalebox{0.6}{3dB}}(\eta^2+1)}$. Hence, 
\begin{subequations}
\begin{equation}
\rf^\mathrm{max} = \frac{\rf \RD \sqrt{ \beta_1 \beta_2} \eta }{\RD \eta \sqrt{ \beta_1 \beta_2} - 4 \rf \alpha_\mathrm{\scalebox{0.6}{3dB}}(\eta^2+1)}, 
\label{eqn_IIIB_6}
\end{equation}
\begin{equation}
 \rf^\mathrm{min} = \frac{\rf \RD \sqrt{ \beta_1 \beta_2} \eta }{\RD \eta \sqrt{ \beta_1 \beta_2} + 4 \rf \alpha_\mathrm{\scalebox{0.6}{3dB}}(\eta^2+1)}.
\label{eqn_IIIB_7}
\end{equation}
\end{subequations} 
The beamdepth $\BD$ is defined as the distance window between $\rf^\mathrm{max}$ and $\rf^\mathrm{min}$, where the array gain $\GR$ is less than or equal to $\unit[3]{dB}$. Hence, the beamdepth $\BD = \rf^\mathrm{max} - \rf^\mathrm{min}$ is given by the expression in (\ref{eqn_IIIB_4}), which concludes the proof. 
\end{proof}

\begin{theorem} 
The maximum range for a given angle at which near-field beamfocusing can be achieved for a \ac{URA} is referred to as the effective beamfocused Rayleigh distance and is expressed as
\begin{equation}
\begin {aligned}
\rf < \frac{\eta\RD }{4\alpha_\mathrm{\scalebox{0.6}{3dB}} (1+\eta^2) }\sin \theta \sqrt{1-\sin^2 \theta \sin^2 \varphi}
\end {aligned}
\label{eqn_IIIB_8}
\end{equation}
\label{theorem2}
\end{theorem}

\begin{proof}
In \eqref{eqn_IIIB_4}, the maximum value of $\BD$ is obtained when the factor in the denominator {\small $ \left[\eta \RD \sin \theta \sqrt{1-\sin^2 \theta \sin^2 \varphi}\right]^2 -\left[4\alpha_\mathrm{\scalebox{0.6}{3dB}} \rf (\eta^2+1)\right]^2 = 0$}. Thus, the farthest angle-dependent axial distance, \(\rf\), within which finite-depth beamforming can be achieved, is given by 
\(\rf < \frac{\eta \RD}{4 \alpha_\mathrm{\scalebox{0.6}{3dB}} (1+\eta^2)} \sin \theta \sqrt{1-\sin^2 \theta \sin^2 \varphi}.\) Otherwise, for distances exceeding this limit, the beamdepth approaches infinity (\(\BD \rightarrow \infty\)).
\end{proof}



The beamdepth and \ac{EBRD} expressions for a \ac{USA} and a \ac{ULA} can be derived as follows:

\newtheorem{corollary}{Corollary}
\begin{corollary}
 For a \ac{USA}, beamdepth $\BD^{\mathrm{\scalebox{0.5}{USA}}}$ is obtained by substituting $\eta=1$ in \eqref{eqn_IIIB_4} to get
\begin{equation}
\BD^{\mathrm{\scalebox{0.5}{USA}}}= \begin{cases}\frac{16 \rf^{2} \alpha_\mathrm{\scalebox{0.6}{3dB}} \sqrt{\beta_1 \beta_2}} {{(\RD \sqrt{\beta_1 \beta_2})} ^{2}- (8 \alpha_\mathrm{\scalebox{0.6}{3dB}} \rf)^{2}}, & \rf<\frac{\RD}{8 \alpha_\mathrm{3dB}} \sqrt{\beta_1 \beta_2}, \\ \infty, & \rf \geq \frac{\RD} {8 \alpha_\mathrm{\scalebox{0.6}{3dB}}} \sqrt{\beta_1 \beta_2}.\end{cases}
\label{eqn_IIIC_1}
\end{equation}
 Note that in the boresight case $ \sqrt{\beta_1 \beta_2} =1$ and $ \alpha_\mathrm{\scalebox{0.6}{3dB}} = 1.25$, then \eqref{eqn_IIIC_1} reduces to the beamdepth expression in [\cite{bjornson2021primer}, Eq. 23].
\end{corollary}
 \begin{corollary}
 For a \ac{ULA} along y-axis, beamdepth $\BD^{\mathrm{\scalebox{0.5}{ULA}}}$ is given by 
\begin{equation} 
\BD^{\mathrm{\scalebox{0.5}{ULA}}} = \begin{cases}\frac{8 \alpha_\mathrm{\scalebox{0.6}{3dB}} \rf^2{\RD} \cos^2{(\varphi)}} {( \RD \cos^2{(\varphi))^2 -(4\alpha_\mathrm{\scalebox{0.6}{3dB}} \rf )^2}},& \rf<\frac{\RD}{4 \alpha_\mathrm{\scalebox{0.6}{3dB}}} \cos^2{(\varphi)}, \\ \infty, & \rf \geq \frac{\RD} {4 \alpha_\mathrm{\scalebox{0.6}{3dB}}}\cos^2{(\varphi)}. \end{cases}
\label{eqn_IIIC_2}
\end{equation}
\begin{proof}
 Ref [\cite{10934779}, Eq. 3].
\end{proof}
\end{corollary}

\begin{figure}[t!]
\centering
\includegraphics[width = 1\linewidth]{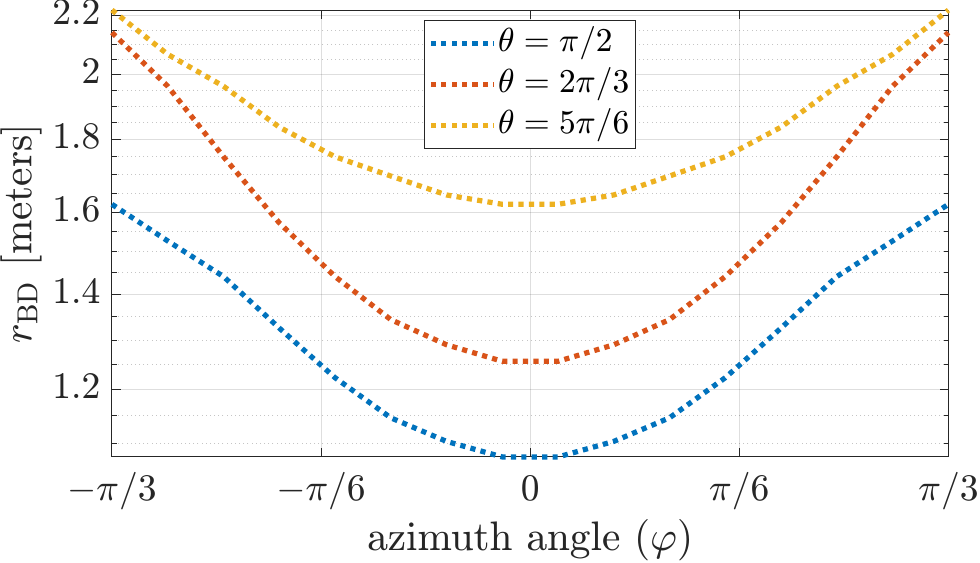}
\caption{The beamdepth of a square URA with respect to azimuth angle where $\NBS=256$, \text{and} $f_c= \unit[28]{GHz}$.}
\label{fig5_BD_vs_angles}
\end{figure}

\subsection{Impact of Spatial Focus Region and Array Geometry} \label{sec-3b}
Beamdepth plays a crucial role in enhancing spatial multiplexing gains within the \ac{EBRD} region. A narrow beamdepth in the near-field minimizes mutual interference, allowing more users to be served in the same angular direction. In the subsequent analysis, we explore how the spatial focus region and array geometry influence beamdepth and \ac{EBRD}.

\subsubsection{Spatial Focus Region} 
As the beam deviates from the boresight, the angular beamwidth expands. Similarly, beamdepth $\BD$ increases as the spatial focus region moves further from the array, exhibiting a quadratic dependence on the focus distance $\rf$. In the angular domain, beamdepth attains its minimum at the boresight ($\varphi = 0, \theta = \pi/2$) and progressively enlarges as the beam direction approaches the endfire directions ($\varphi = \pi/2, \theta = \pi$). For wide arrays ($\eta \gg 1$), the azimuth angle $\varphi$ predominantly influences beamdepth, whereas for taller arrays ($\eta \ll 1$), the elevation angle $\theta$ plays a more significant role. In general, beamdepth increases as either the azimuth or elevation angle deviates from the boresight. For instance, Fig. \ref{fig5_BD_vs_angles} illustrates the variation in beamdepth with respect to the azimuth angle $\varphi$ for a \ac{USA}, demonstrating an increase in $\BD$ with larger azimuth angles. Likewise, as the elevation angle $\theta$ transitions from the boresight at $\pi/2$ toward the endfire direction at $5\pi/6$, beamdepth $\BD$ also increases.

\subsubsection{Array Geometry} We aim to determine the optimal width-to-height ratio $\eta$ for a \ac{URA} that minimizes the beamdepth $\BD$ in \eqref{eqn_IIIB_4}, given a fixed number of antenna elements $\NBS$. The parameter $\eta$ not only appears explicitly in \eqref{eqn_IIIB_4} but also influences other critical factors, such as the Rayleigh distance $\RD$ and the coefficient $\alpha_\mathrm{\scalebox{0.6}{3dB}}$.

For a fixed number of antennas $\NBS$, the Rayleigh distance $\RD = \frac{2D^2}{\lambda}$ increases as the URA becomes more elongated, either in height or width, due to the corresponding increase in aperture length $D$. We simulate a range of tall and wide array configurations by varying $\eta$ from $2^{-6}$ to $2^{6}$, as illustrated in Fig. \ref{fig7_BD_Array}. The results indicate that the Rayleigh distance $\RD$ reaches its minimum when $\eta = 1$, corresponding to a square \ac{URA}.

Furthermore, the factor $\alpha_\mathrm{\scalebox{0.6}{3dB}}$, as defined in (\ref{eqn_IIIB_5}), is influenced by $\eta$. The gain function $\GR$ in \eqref{eqn_IIIB_3} exhibits symmetry with respect to $\eta$, satisfying the relation $\GR(\eta) = \GR(1 / \eta)$. Consequently, $\alpha_\mathrm{\scalebox{0.6}{3dB}}$ is also symmetric around $\eta = 1$, as illustrated in Fig. \ref{fig7_BD_Array}. Additionally, the dependence of beamdepth $\BD$ on $\eta$, given in \eqref{eqn_IIIB_4}, can be approximated by $\frac{\eta^2 + 1}{\eta}$, which is also depicted in Fig. \ref{fig7_BD_Array}. Among different array configurations, the combined factor $\alpha_\mathrm{\scalebox{0.6}{3dB}}\left(\frac{\eta^2 + 1}{\eta}\right)$ reaches its highest value for a square \ac{URA} compared to more elongated structures.

While the combined factor $\alpha_\mathrm{\scalebox{0.6}{3dB}}\left(\frac{\eta^2 + 1}{\eta}\right)$ is maximized for a square \ac{URA}, the Rayleigh distance $\RD$ is at its lowest in this configuration. Among these two factors, $\RD$ has a greater influence on beamdepth $\BD$ due to its larger magnitude. As a result, the beamdepth $\BD$ is minimized for a square \ac{URA} and becomes larger for other array geometries.

\begin{figure}[t!]
\centering
\includegraphics[width = 1\linewidth]{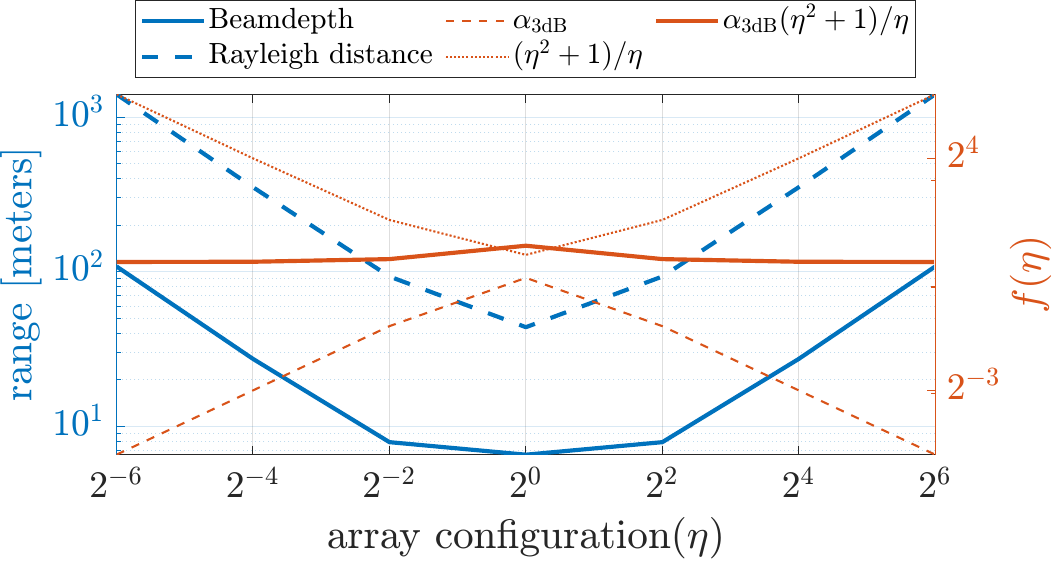}
\caption{Left axis: Variation of beamdepth and Rayleigh distance as a function of \(\eta\). Right axis: The factor \(\alpha_\mathrm{3dB}\left(\frac{\eta^2 + 1}{\eta}\right)\), derived from \eqref{eqn_IIIB_4}, plotted against \(\eta\). We set $\NBS=4096$, \text{and} $f_c= \unit[28]{GHz}$.}
\label{fig7_BD_Array}
\end{figure}

\subsubsection{Limits of Beamfocusing}
The \ac{EBRD} represents the boundary beyond which near-field beamfocusing is no longer achievable. As formulated in \eqref{eqn_IIIB_8}, its value is determined by the Rayleigh distance $\RD$ and the focus direction. For a given array, increasing the aperture size $D$—either by adding more antenna elements $\NBS$ or enlarging the inter-element spacing $d$—results in an extended \ac{EBRD}. Moreover, deploying widely spaced sub-arrays offers a means to achieve a large \ac{EBRD} without requiring a significant increase in the number of antennas \cite{kosasih2024achieving}.

Since a square \ac{URA} yields the smallest Rayleigh distance $\RD$, it also results in the minimum \ac{EBRD}, which increases as the array shape transitions toward a \ac{ULA}. This pattern is illustrated in Fig. \ref{fig8_EBRD}, where the \ac{EBRD} is plotted against the azimuth angle $\varphi$ for different values of $\eta$ and elevation angles $\theta$. For wide arrays, variations in the \ac{EBRD} are primarily influenced by the azimuth angle, whereas for tall arrays, the elevation angle plays a dominant role. Regardless of the configuration, the \ac{EBRD} reaches its maximum at the boresight direction ($\varphi = 0, \theta = \pi/2$) and decreases toward the endfire directions.

Among various array configurations, the maximum \ac{EBRD} occurs when $\eta = 0.004$, whereas the minimum \ac{EBRD} is observed for a square \ac{URA} ($\eta = 1$). For tall \acp{URA} with $\eta \ll 1$, such as $\eta = 0.016$ and $\eta = 0.004$, the \ac{EBRD} exhibits minimal variation with respect to the azimuth angle $\varphi$. However, as the elevation angle shifts away from the boresight ($\theta = \pi/2$), the \ac{EBRD} progressively decreases.

To summarize, the square \ac{URA} yields the smallest beamdepth $\BD$, which facilitates superior spatial multiplexing capabilities. Nevertheless, its effective near-field coverage, as determined by the \ac{EBRD}, is very limited. On the other hand, wide or tall \ac{URA} configurations offer a larger \ac{EBRD}, albeit at the expense of increased beamdepth relative to the square configuration. A detailed numerical evaluation of the multiuser capacity is presented in Section \ref{sec-4} to determine the array geometry that achieves the highest sum rate.

\begin{figure}[t!]
\centering
\includegraphics[width = 1\linewidth]{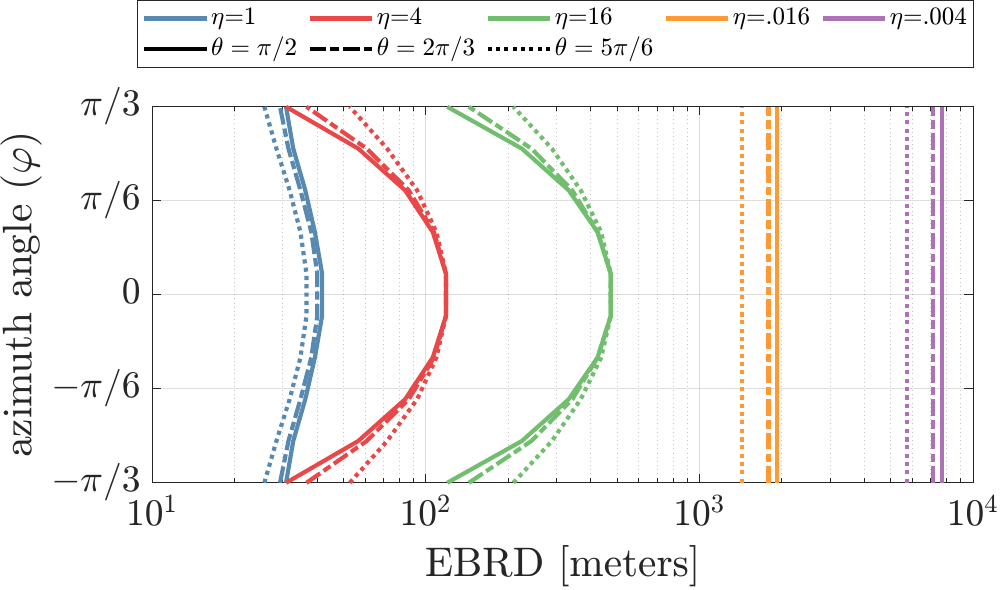}
\caption{EBRD as a function of azimuth angle, where $\eta = 1$ represents a square \ac{URA}, $\eta =\{4, 16\}$ correspond to wide \acp{URA}, $\eta = 0.016$ denotes a very tall \ac{URA}, and $\eta = 0.004$ represents a \ac{ULA}. We set $\NBS = 4096$ and $f_c = \unit[28]{GHz}$.}
\label{fig8_EBRD}
\end{figure}

\section{Numerical Results} \label{sec-4}
We conduct Monte Carlo simulations to evaluate the achievable sum rate across different array configurations. The simulation setup considers a downlink communication system, where a \ac{BS} serves $\MUE$ users. The \ac{BS} employs $\NBS$ antenna elements and $\NRF=4$ RF chains, utilizing hybrid precoding. To estimate the channel, the \ac{BS} performs beam training through beam sweeping. Each \ac{UE} reports the index of the codeword that provides the highest gain, enabling the \ac{BS} to select the optimal codeword from the polar codebook for analog precoding, represented by $\mathbf{W}$. Finally, the \ac{BS} applies zero-forcing for digital precoding. We compute the overall \ac{SE} using the following expression
\begin{equation}
R = \sum_m \log_2 \left( 1 + \frac{ p_m \left| \mathbf{h}_{m}^{\mathsf{H}} \mathbf{W} \mathbf{f}_{m} \right|^2}{\sigma_m^2 + \sum_{l \neq m} p_l\left| \mathbf{h}_{m}^{\mathsf{H}} \mathbf{W} \mathbf{f}_{l} \right|^2} \right),
\label{eqn34}
\end{equation}
where $p_m$ denotes power allocated to the $\nth{m}$ user and $\sigma_m^2$ denotes the noise variance. Furthermore, $f_m$ represents the $\nth{m}$ column of the digital precoder.

The finite beamdepth remains achievable only within the \ac{EBRD} region, constraining the spatial multiplexing gains in the near-field domain defined by the \ac{EBRD}. To illustrate this effect, we examine a $64 \times 8$ \ac{URA} serving five users positioned along the boresight direction $(\varphi = 0, \theta = \pi/2)$. These users are distributed across three distinct regions: the \ac{EBRD} region $[1.2D, \EBRD]$, the extended near-field beyond \ac{EBRD} $[\EBRD, \RD]$, and the far-field region $[\RD, 100\RD]$. We evaluate performance using both polar and \ac{DFT} codebooks across these regions. In the polar codebook, codewords are sampled across the EBRD region as described in \cite{10988573}, whereas in the DFT codebook, codewords are uniformly sampled across the angular region. 
The spectral efficiency results, presented in Fig. \ref{fig15_SE_vs_EBRD}, show that within the \ac{EBRD} region, the polar codebook achieves a spectral efficiency of $5.8$ bps/Hz, significantly surpassing the \ac{DFT} codebook, which only reaches $2$ bps/Hz. However, beyond the \ac{EBRD}, in both the extended near-field and far-field regions, the spectral efficiency of both codebooks remains comparable. These findings highlight the \ac{EBRD} as a critical threshold distinguishing the radiative near-field from the far-field in terms of achievable multiplexing gains.

The array geometry plays a crucial role in determining beamdepth and the boundaries of the \ac{EBRD}. A square \ac{URA} $(\eta = 1)$ achieves the smallest beamdepth, which enhances spatial multiplexing. However, this configuration also results in a more constrained \ac{EBRD}, potentially limiting overall system capacity when compared to wider or taller array structures. To demonstrate this effect, we analyze two configurations: a square \ac{URA} with dimensions $[32 \times 32]$ $(\eta = 1)$ and a wide \ac{URA} with dimensions $[128 \times 8]$ $(\eta = 16)$. Both arrays contain the same number of antenna elements, $\NBS = 1024$, and serve six users aligned along the boresight direction. Fig. \ref{fig16_SE_vs_eta} presents the \ac{SE} as a function of \ac{SNR} for these configurations. The results indicate that the wide \ac{URA} $(\eta = 16)$ consistently outperforms the square \ac{URA} $(\eta = 1)$ across all \ac{SNR} levels, demonstrating that wider array designs offer enhanced \ac{SE}. At high \ac{SNR}, the wide \ac{URA} attains a \ac{SE} of $\unit[8.2]{bps/Hz}$, significantly outperforming the square \ac{URA}, which achieves only $\unit[2.3]{bps/Hz}$.

\begin{figure}[t!]
\centering
\includegraphics[width = 1\linewidth]{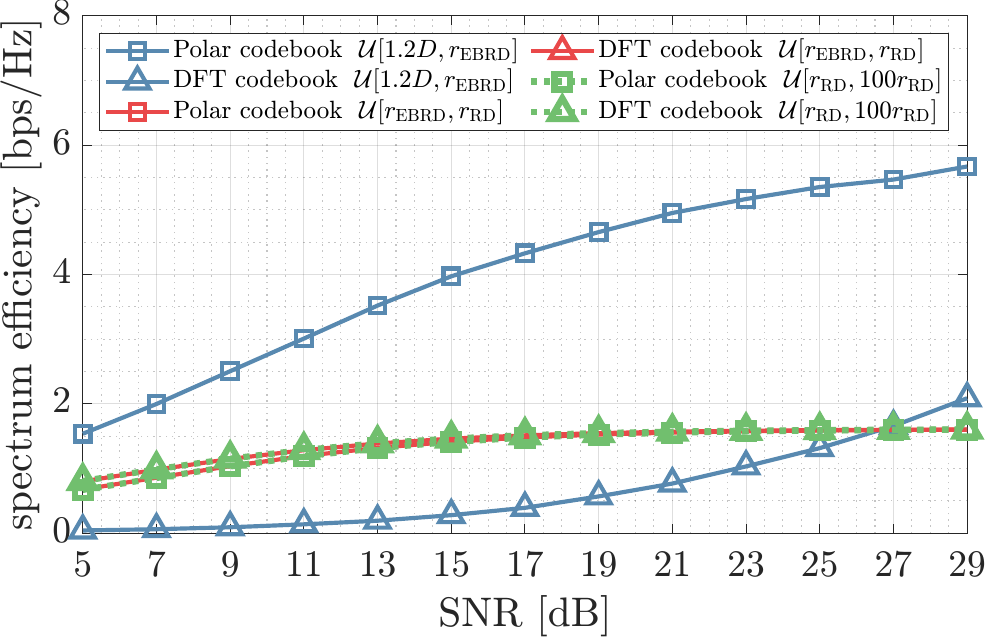}
\caption{ Spectral efficiency versus SNR using polar codebook and DFT codebooks.}
\label{fig15_SE_vs_EBRD}
\end{figure}

\begin{figure}[t]
\centering
\includegraphics[width = 1\linewidth]{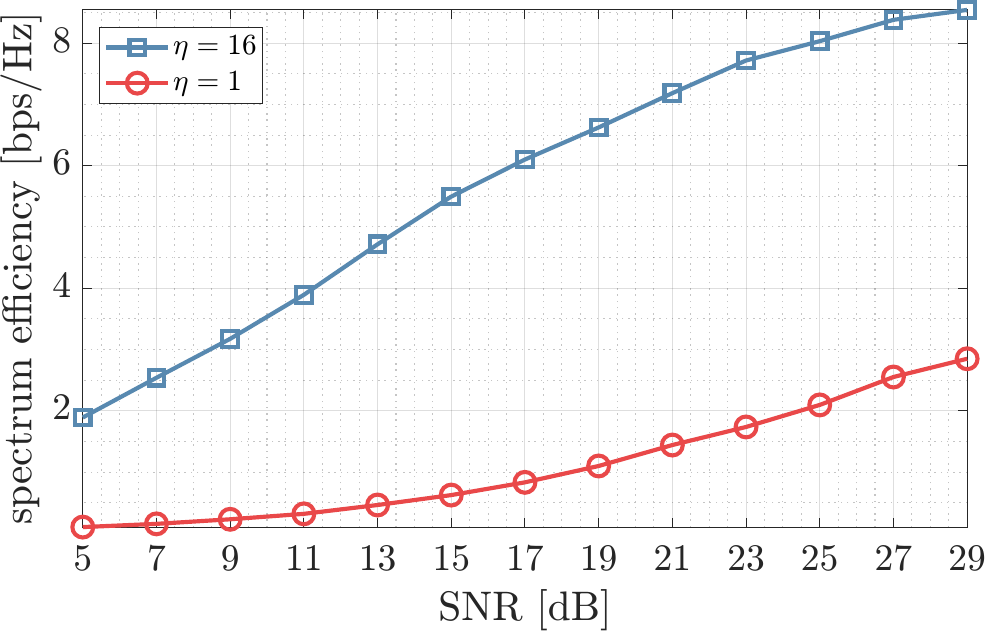}
\caption{ Spectral efficiency versus SNR for different URA configurations.}
\label{fig16_SE_vs_eta}
\end{figure}

\section{Conclusion} \label{sec-5}
In this paper, we have analyzed the beamdepth and \ac{EBRD} for a generalized \ac{URA}. Specifically, we have derived the beamdepth by evaluating the normalized array gain at the $\unit[3]{dB}$ points and introduced the \ac{EBRD} to define the near-field region where beamfocusing is achievable. Our analysis has demonstrated that the beamdepth was narrowest at boresight and increased at off-boresight angles, while the \ac{EBRD} was maximized at boresight and decreased at off-boresight angles. To assess the impact of array geometry, we have shown that the beamdepth is minimized for a square \ac{URA}, whereas the \ac{EBRD} is maximized for wide or tall \ac{URA} configurations.

\bibliographystyle{IEEEtran}
\bibliography{Bibliography/IEEEabrv,Bibliography/my2bib}
\end{document}